\documentclass[conference]{IEEEtran}
\IEEEoverridecommandlockouts

\usepackage{amssymb,latexsym,amsfonts,amsmath,amsthm}
\usepackage{bbm}
\usepackage{multicol}
\usepackage{multirow}
\usepackage{algorithmic}
\usepackage{cite}
\usepackage{graphicx}
\usepackage{enumerate}
\usepackage{url}
\usepackage{caption}
\usepackage{textcomp}
\usepackage{xcolor}
\usepackage{dsfont}

\def\BibTeX{{\rm B\kern-.05em{\sc i\kern-.025em b}\kern-.08em
    T\kern-.1667em\lower.7ex\hbox{E}\kern-.125emX}} 

\newtheorem{theorem}{Theorem}
\newtheorem{lemma}{Lemma}
\newtheorem{definition}{Definition}
\newtheorem{assumption}{Assumption}

\newtheorem{example}{Example}[section]

\begin{document}

\title{\LARGE {\bf Synthesis of Discounted-Reward Optimal Policies for Markov  Decision  Processes Under  Linear  Temporal  Logic  Specifications}}
\author{\normalsize Krishna C. Kalagarla, Rahul Jain,  Pierluigi Nuzzo\\
Ming Hsieh Department of Electrical and Computer Engineering, University of Southern California, Los Angeles \\
Email: \{kalagarl,rahul.jain,nuzzo\}@usc.edu
}

\maketitle

\begin{abstract}
 We present a method to find an optimal policy with respect to a reward function for a discounted Markov decision process under general linear temporal logic (LTL) specifications. Previous work has either focused on maximizing a cumulative reward objective under finite-duration tasks, specified by syntactically co-safe LTL, or maximizing an average reward for persistent (e.g., surveillance) tasks. This paper extends and generalizes these results by introducing a pair of occupancy measures to express the LTL satisfaction objective and the expected discounted reward objective, respectively. These occupancy measures are then connected to a single policy via a novel reduction resulting in a mixed integer linear program whose solution provides an optimal policy. Our formulation can also be extended to include additional constraints with respect to secondary reward functions. We illustrate the effectiveness of our approach in the context of robotic motion planning for complex missions under uncertainty and performance objectives. 
\end{abstract}

\section{Introduction}

The deployment of autonomous systems in safety critical applications, such as transportation \cite{koopman2017autonomous}, robotics \cite{guiochet2017safety}, and advanced manufacturing, has been rapidly increasing over the past few years, 
calling for high-assurance design methods that can provide strong guarantees of system correctness, safety, and performance. 


Markov decision processes ~\cite{Puterman:1994:MDP:528623} offer a natural
framework to capture sequential decision making problems and reason about autonomous system behaviors in many applications where the system dynamics are not deterministic but rather affected by uncertainty. 
In classical MDP planning, rewards are assigned to pairs of states and actions in the MDP. An expected (total, discounted, or averaged) reward is then maximized to find an optimal policy, which ensures that a design objective is achieved \cite{sutton2018reinforcement,mnih2013playing, kaelbling1996reinforcement}. A critical step in this process is the formulation of the reward structure, since an incorrect definition may lead to unforeseen behaviors that can be unsafe or even fail to meet the requirements. To alleviate this difficulty, an increasing interest has been directed over the past decade toward tools from formal methods and temporal logics \cite{baier2008principles} as a way to unambiguously capture complex design objectives and rigorously validate critical requirements  like safety, task scheduling, and motion planning in the context of discrete transition systems \cite{smith2011optimal, fainekos2009temporal, ulusoy2013optimality}, hybrid systems \cite{wongpiromsarn2011tulip, fainekos2005hybrid , wongpiromsarn2012receding}, and MDP planning~\cite{lahijanian2010motion, ding2014optimal,lahijanian2011temporal}.


Logically-driven planning for MDPs can be formulated as the problem of finding a policy that maximizes the probability of satisfying a given temporal logic formula~\cite{ding2011ltl,hahn2019omega, bozkurt2019control,hasanbeig2019certified}. For example, linear temporal logic (LTL) is expressive enough to formally capture, among others, safety, progress, surveillance, and monitoring tasks, such as ``avoid region A'' (safety), ``eventually reach A and remain there forever'' (reachability), and ``infinitely often visit A'' (surveillance). However, while certain design constraints, e.g., related to system performance, smoothness of motion, or fuel consumption rates, are deemed as ``soft'' in many applications, and can be naturally expressed in terms of expected reward or probability maximization, certain ``hard'' system constraints, such as safety and mission critical requirements, call for stronger guarantees, possibly involving almost sure satisfaction of temporal logic objectives. The focus of this  paper is on these composite tasks where a reward objective (characterizing a ``soft,'' user-defined objective) on an MDP is maximized while attempting to guarantee satisfaction of a temporal logic formula (characterizing a ``hard,'' mission-related objective) with probability one.

MDP planning under LTL constraints and reward optimality requirements has attracted significant attention in the recent literature. Some  efforts~\cite{lacerda2014optimal,lacerda2015optimal,ding2013strategic} aim to optimize a total expected reward under the constraint that a co-safe LTL formula is satisfied with high probability. Co-safe LTL is a fragment of LTL capable of expressing properties that can be satisfied over a finite time horizon. It is, however, less effective for the specification of persistent tasks such as surveillance or monitoring. The satisfaction of optimality criteria for persistent tasks  has been the focus of other formulations~\cite{ding2014optimal,wu2017learning}, where the objective is formalized as the expected average reward over trajectories satisfying an LTL formula. However, expected average rewards tend to emphasize optimality in the 
steady state and can possibly neglect important behaviors during the transient phase of a mission. A recent result aims to combine these approaches via a customized,  user-defined reward objective, defined as the weighted sum of an expected \emph{total reward} associated with the trajectory prefix and an expected \emph{average reward} for the trajectory suffix~\cite{guo2018probabilistic}. A \emph{unifying framework for  total (discounted) reward optimality} that extends to both trajectory prefixes (transient behaviors) and suffixes (steady-state behaviors) and can support unbounded-horizon tasks 
specified by \emph{general LTL} has been missing.  

This paper bridges this gap by proposing a new method that leverages the notion of occupancy measures~\cite{altman1999constrained} to generate a policy for a given MDP that almost surely satisfies a general LTL specification and maximizes the expected discounted reward. By optimizing an expected discounted reward, we
can adequately account for important transient behaviors as well as long-term optimality in a single objective. In our framework, discounted reward optimality and LTL satisfiability are both translated into linear constraints on a pair of occupancy measures, which we show can be connected to a single policy via additional binary variables, leading to  a mixed integer linear program. Our formulation can also be extended to include additional constraints expressed in terms of minimum expected discounted return with respect to secondary reward functions.
We illustrate the applicability of our approach on case studies in the context of robotic motion planning, showing that it can effectively discriminate optimal policies according to different expected discounted rewards among the ones that almost surely satisfy an LTL specification.

The rest of the paper is organized as follows. After an overview of the related work in Section \ref{sec:related}, Section \ref{sec:prelim} describes the preliminaries and Section \ref{sec:milp} introduces the problem formalization and elaborates on the solution approach. Finally, Section \ref{sec:experiment} discusses the application of the framework on two case studies in the context of motion planning.

\section{Related Work} \label{sec:related}

By building on results from model checking~\cite{baier2008principles}, the synthesis of MDP policies that maximize the probability of satisfaction of an LTL formula relying on maximizing     
the probability of reaching the maximal accepting end components of the product automaton between the original MDP and a (Rabin or B\"{u}chi) automaton representing the LTL formula has been studied in the context of known transition probabilities~\cite{ding2011ltl,sickert2016limit} and unknown transition probabilities \cite{hahn2019omega,bozkurt2019control,fu2014probably,hasanbeig2019certified,hasanbeig2020cautious}. The synthesis of a policy whose probability of satisfying multiple LTL specifications is above certain thresholds has also been considered in the literature~\cite{etessami2007multi}, albeit without additional quantitative reward objectives.

A method to synthesize a policy which maximizes the probability of satisfaction of an LTL formula via reward shaping has been proposed based on a Rabin acceptance condition ~\cite{sadigh2014learning,hiromoto2015learning}.
However, finding a policy that almost surely satisfies the LTL formula is not guaranteed, even if such a policy exists, when multiple Rabin pairs or rejecting end components are present~\cite{hahn2019omega}. Representing the LTL formula via a limit deterministic B\"{u}chi automaton (LDBA) has been recently proposed as an effective alternative~\cite{hahn2019omega,hasanbeig2019certified,bozkurt2019control,sickert2016limit}. LDBAs are as expressive as deterministic Rabin  automata (DRAs), while offering simpler B\"{u}chi acceptance conditions (requiring visiting a set of states infinitely often) rather than Rabin conditions (requiring visiting a set of states infinitely often and another set of states finitely often). This paper builds on the latter approach and further extends it to synthesize optimal policies that maximize an additional reward function while satisfying an LTL formula almost surely. 

Existing approaches to synthesizing reward-optimal policies under LTL specifications have focused on long-term average cost on states which must be visited infinitely often~\cite{ding2011mdp,ding2014optimal}, motivated by surveillance tasks, or total cost under co-safe (bounded-time horizon) LTL specifications~\cite{ding2013strategic,lacerda2014optimal,lacerda2015optimal}. The former approach relies on extending the average-cost-per-stage problem~\cite{bertsekas1995dynamic} and utilizing the DRA acceptance condition to formulate a dynamic programming problem. The latter efforts use a single occupancy measure to capture the cost and the probability of satisfying the co-safe LTL formula. 
Similarly, multiple reachability objectives and a total undiscounted cost have been considered under additional restrictions, where zero-reward goal states must be reached almost surely to ensure that the total cost is  finite~\cite{forejt2011quantitative,hartmanns2018multi}. In these cases, a single occupancy measure can be used to express the total cost and the reachability objectives, but it would not be sufficient to capture discounted objectives. Multiple discounted objectives (albeit with the same discount factor) can be expressed
via a single occupancy measure~\cite{chatterjee2006markov}, but this would not extend to capture an LTL objective.
More recently, Guo~et~al.~\cite{guo2018probabilistic} propose to account for both transient and steady-state behaviors by defining a customized cost as a weighted sum of a prefix-related total cost and a suffix-related average cost. It is then possible to optimize for this objective by introducing two occupancy measures for the trajectory prefix (to account for the probability of reaching the accepting end components and the associated total cost) and its suffix (to account for the average cost within the end components), respectively. 

Our approach differs from all these efforts, in that it can deal with a single, total discounted reward, without any additional assumptions on the real-valued reward function, and generic LTL specifications. By focusing on a total discounted reward, our approach  does not require the additional restriction that the total cost be finite. To enable these extensions under a unified reward structure, 
our framework must also rely on the consistency between two occupation measures. However, differently from the proposal above~\cite{guo2018probabilistic}, our measures are used to account for the \emph{total discounted reward} and the \emph{generic LTL satisfiability} objectives, respectively, and they are defined \emph{over the entire trajectory}.
 
By targeting almost sure satisfaction of generic LTL formulae via a single mixed integer program, our approach differs from previous efforts that leverage parity games~\cite{ufukgame} to achieve $\epsilon$-optimal policies for discounted rewards and a subclass of LTL formulae, or provide Pareto efficient policies with respect to a discounted reward and discounted probability of satisfying an LTL specification~\cite{fu2015pareto}.

Finally, mixed integer linear program (MILP) formulations were considered in the past to compute counterexamples to reachability properties and expected-reward constraints~\cite{quatmann2015counterexamples}, permissive schedulers under safety constraints~\cite{junges2016safety}, and Pareto fronts for multiple total reward and reachability objectives~\cite{delgrange2020simple}. We focus, instead, on optimal policy synthesis for discounted rewards and formulate a MILP over the occupancy measures, rather than the value functions.

\section{Preliminaries} \label{sec:prelim}

We denote the sets of real and natural numbers by $\mathbb{R}$ and $\mathbb{N}$, respectively. $\mathbb{R}_{\geq 0}$ is the set of non-negative reals. For a given finite set $S$, $S^{\omega}$ ($S^{+}$) denotes the set of all infinite (finite) sequences taken from $S$. The indicator function $\mathds{1}_{s_0}(s)$ evaluates to $1$ when $s = s_0$ and 0 otherwise. 

\noindent \textbf{Markov Decision Process.} 
A (labeled) Markov Decision Process (MDP) is defined as a tuple $ \mathcal{M} = (S,A,P,s_0,\gamma,AP,L,r)$, where $S$ is a finite state space, $A$ is a finite action space, $P: S \times A \times S \to [ 0,1 ]$ is the partial transition probability function, such that $P(s,a,s')$ is the probability to transition from state $s$ to state $s'$ on taking action $a$, $s_0 \in S $ is the initial state, $\gamma \in (0,1)$ is the discount factor, $AP$ is a set of atomic propositions, 
$L: S \to 2^{AP}$ is a labeling function which indicates the set of atomic propositions which are true in each state, and $r: S \times A \to \mathbb{R}$ is a reward function, such that $r(s,a)$ is the reward obtained on taking action $a\in A$ in state $s \in S$. We let $A(s)$ denote the set of actions allowed in state $s \in S$. 

The MDP evolves starting from the initial state $s_0$ by taking action $a \in A(s)$ for every current state $s$. A finite \emph{run} ${\xi_t}$ of the MDP at time $t \in \mathbb{N}$ is a sequence of past states and actions $s_0a_0s_1,a_1\ldots s_{t}$ up to time $t$. An infinite run $\xi$ is obtained by letting $t$ tend to infinity.  

A policy $\pi$ is a sequence of decision functions $\pi_0\pi_1\ldots$ such that  $\pi_t$ maps $\xi_t$ to the set of actions $A(s_t)$ available to the state at time $t$. 
If $\pi_t = \pi$ for all $t$, then the policy $\pi$ is said to be \emph{stationary}. A stationary policy $\pi$ is  \emph{randomized} when it is a probability distribution over the available actions, i.e.,  $\pi : S \times A \to [0,1] $. It is, instead,  \emph{deterministic} when it provides a unique action for each state, i.e., $\pi : S \to A$. On the other hand, a policy is said to be \emph{finite memory} if it also depends on the history of ${\xi_t}$ in addition to the current state. 

Given the total discounted reward associated with a run $\xi$ of $\mathcal{M}$, defined as $ G_{\mathcal{M}}(\xi) =  \sum_{t=0}^{\infty} { \gamma }^{t} r(s_t,a_t)$, we define the expected discounted reward as follows. 

\begin{definition}[Expected Discounted Reward] 
The expected discounted reward $\mathcal{R}_{\mathcal{M}}({\pi})$ of policy $\pi$ for MDP $ \mathcal{M}$ is the expectation of the total discounted reward obtained by following policy $\pi$ on $\mathcal{M}$, i.e.,
%
     $ \mathcal{R}_{\mathcal{M}}({\pi}) = {\mathbb E}_{\pi} \left[ G_{\mathcal{M}}(\xi) \right]$,
%
where the expectation is taken over the probability measure induced by policy $\pi$ on $\mathcal{M}$.
\end{definition}

\noindent\textbf{Linear Temporal Logic.}
We use linear temporal logic (LTL)~\cite{baier2008principles}, a temporal extension of propositional logic, to express complex task specifications. 

\vspace{3pt}
\noindent\underline{Syntax.} Given a set $AP$ of atomic propositions, i.e., Boolean variables that have a unique truth value ($\mathsf{true}$ or $\mathsf{false}$) for a given system state, LTL formulae are constructed inductively as  follows: 
\begin{equation*}
    \varphi := \mathsf{ true } \ | \ a \ | \ \neg  \varphi \ | \ \varphi_1 \wedge \varphi_2 \ | \ \textbf{X} \varphi \ | \ \varphi_1 \textbf{U} \varphi_2, 
\end{equation*}
where $a \in AP$, $\varphi$, $\varphi_1$, and $\varphi_2$ are LTL formulae, $\wedge$ and $\neg$ are the logic conjunction and negation,  and $\textbf{U}$ and $\textbf{X}$ are the \emph{until} and \emph{next} temporal operators. Additional temporal operators such as \emph{always} ($\textbf{G}$) and \emph{eventually} ($\textbf{F}$) are derived as $\textbf{F} \varphi := \mathsf{ true } \textbf{U} \varphi$ and $\textbf{G} \varphi := \neg \textbf{F} \neg \varphi$.

\vspace{3pt}
\noindent\underline{Semantics.} 
LTL formulae are interpreted over infinite-length words $w = w_0w_1 \ldots \in {(2^{AP})}^{\omega}$, where each letter $w_i$ is a set of atomic propositions. $w[j \ldots ] = w_jw_{j+1} \ldots $ is the suffix of $w$ starting from letter $w_j$. Informally, a word $w$ satisfies $a$ if $a \in w_0$ holds; it satisfies $\textbf{X} \varphi$ if $w[1 \ldots ]$ satisfies $\varphi$; it satisfies $\varphi_1 \textbf{U} \varphi_2$ if there exists $k\geq 0$ such that $w[k \ldots ]$ satisfies $\varphi_2$ and for all $j \in \{0,\ldots,k-1\}$, $w[j \ldots ]$ satisfy $\varphi_1$. The word $w$ satisfies $\textbf{F} \varphi$ if there exists $j\geq 0$ such that $w[j \ldots ]$ satisfies $\varphi$. Finally, $w$ satisfies $\textbf{G} \varphi$ if $w[j \ldots ]$ satisfies $\varphi$ for all $j\geq 0$. These semantics are formally defined as follows, where $\models$ denotes satisfaction: 
\begin{align*}
    w &\models \mathsf{ true },\\
    w &\models a \text{ if and only if  } a \in w_0,\\
    w &\models \varphi_1 \wedge \varphi_2 \text{ if and only if  } w \models \varphi_1 \text{ and } w \models \varphi_2,\\
    w &\models \neg \varphi \text{ if and only if  } w \not \models \varphi,\\
    w &\models \textbf{X} \varphi \text{ if and only if  } w[1\ldots] = w_1w_{2} \ldots \models  \varphi, \\
    w &\models \varphi_1 \textbf{U}  \varphi_2\text{ if and only if  } \exists \ k \geq 0 \text{ s.t. } w[k \ldots] \models \varphi_2\\ \text{ and } & \forall j  < k , w[j \ldots] \models \varphi_1,\\
    w &\models \textbf{G} \varphi \text{ if and only if  } \forall \ i \geq 0 , w[i \ldots] \models \varphi,\\
    w &\models \textbf{F} \varphi\text{ if and only if  } \exists \ i \geq 0 , w[i \ldots] \models \varphi.
\end{align*}

Given an MDP $\mathcal{M}$ and an LTL formula $\varphi$, a run $\xi = s_0a_0s_1a_1\ldots$ of the MDP under policy $\pi$ is said to satisfy $\varphi$ if the
word $w = L(s_0)L(s_1)\ldots \in {(2^{AP})}^{\omega}$ generated by the run satisfies $\varphi$. The probability that a run of $\mathcal{M}$ satisfies $\varphi$ under policy $\pi$ is denoted by $Pr_{\pi}^{\mathcal{M}}(\varphi)$. 

\noindent\textbf{Limit Deterministic B\"{u}chi Automata.} 
The language defined by an LTL formula, i.e., the set of words satisfying the formula, can be captured by a Limit Deterministic B\"{u}chi Automaton (LDBA)~\cite{hahn2019omega,kvretinsky2018rabinizer}. 

We consider an LDBA $\mathcal{A}$ denoted by a tuple $(Q, \Sigma, q_0, \delta, F)$, where $Q$ is a finite set of states, $\Sigma$ is a finite alphabet, $q_0 \in Q$ is an initial state, $\delta :
Q \times (\Sigma \cup \{\epsilon\})  \to 2^ {Q}$  is a partial transition function, and $F \subseteq Q \times \Sigma \times Q$ is a set of accepting transitions. We denote by  $\epsilon$-transition a state transition that is not labelled (triggered) by a symbol in $\Sigma$. $\mathcal{A}$ can be seen as two disjoint deterministic automata connected by $\epsilon$-transitions~\cite{hahn2019omega,kvretinsky2018rabinizer}. Specifically, the states $Q$ can be partitioned into a set of initial states $Q_i$ and a set of accepting states $Q_f$ such that: (i) $Q_i$ and  $Q_f$ are disjoint, i.e., $Q_i \cap Q_f = \emptyset$; (ii) $Q_i$ and  $Q_f$ are connected only by $\epsilon$-transitions; and (iii) the following properties hold: 
\begin{itemize}
\item $\epsilon$-transitions are allowed only from $Q_i$ to $Q_f$, i.e., $|\delta(q,\epsilon)| = \emptyset, \ \forall q \in Q_f$, and $\delta(q,\epsilon) \cap Q_i = \emptyset, \ \forall q \in Q_i $;
\item all transitions except $\epsilon$-transitions are deterministic, i.e., $ |\delta(q,\alpha)| \leq 1,\  \forall q \in Q, \forall \alpha \in \Sigma$;
\item all transitions starting from $Q_f$ end in $Q_f$, i.e., $\delta(q,\alpha) \subseteq  Q_f, \  \forall  q \in Q_f, \forall \alpha \in \Sigma$;
\item all accepting transitions lie in $Q_f$, i.e., $F \subseteq Q_f \times \Sigma \times Q_f $.
\end{itemize}
Therefore, the only non-deterministic transitions are $\epsilon $-transitions from $Q_i$ to $Q_f$. 

A \emph{run} $\xi$ of an LDBA $\mathcal{A}$ on word $w \in \Sigma^{\omega}$ is an infinite sequence of transitions,   $\xi_{0}\xi_{1}\ldots = (q_0,w_0,q_1)(q_1,w_1,q_2) \ldots$, 
 such that $\xi_i \in \delta$, $\forall i \geq 0$. The set of transitions that occur infinitely often in $\xi$ are denoted by $\inf(\xi)$.
A run $\xi$ is \emph{accepting} if  $\inf(\xi) \cap F \neq \emptyset$, i.e., the B\"{u}chi acceptance condition holds. In this case, we also use the formula $\varphi_F$ to express the acceptance condition and say that $\xi \models \varphi_F$. A word $w \in \Sigma^{\omega}$ is accepted by $\mathcal{A}$ if and only if there exists an accepting run $\xi$ of $\mathcal{A}$ on $w$. Finally, we say that an LTL formula is equivalent to an LDBA $\mathcal{A}$ if and only if the language defined by the formula is the language accepted by $\mathcal{A}$.

In this paper, we will refer to this kind of LDBA representation, which can be obtained, for example, by first constructing a Limit Deterministic Generalized B\"{u}chi Automaton (LDGBA)~\cite{sickert2016limit}, characterized by multiple acceptance sets $F_i$, $i = 1,\ldots,n$, each of which must be visited infinitely often for a run to be accepted. The LDGBA can then be reduced to an equivalent LDBA as defined above, by leveraging the procedure outlined below, similar to the one used to reduce a Generalized B\"{u}chi Automaton (GBA) to an equivalent Non-deterministic B\"{u}chi Automaton (NBA)~\cite{baier2008principles}.

Given an LDGBA with number of acceptance sets $n > 1$, all belonging to $Q_f$ by construction~\cite{sickert2016limit}, we can replicate the set of states $Q_f$ and the transitions among them $n-1$ times to obtain $Q_{f1}, \ldots, Q_{fn}$. We can then modify the accepting transitions $F_1,\ldots, F_n$ in $Q_{f1}, \ldots, Q_{fn}$, respectively, such that, for $i= 1,\ldots,n-1$, each accepting transition $(q,e,q')$ of $F_i$ in $Q_{fi}$ is replaced by a new transition that starts in $Q_{fi}$ and ends in the state that corresponds to $q'$ in  $Q_{f(i+1)}$. Similarly, the accepting transitions $F_n$ of $Q_{fn}$ can be modified to end in $Q_{f1}$. The newly defined transitions corresponding to $F_{1}$ of $Q_{f1}$, which we denote by $F'$, form the single acceptance set of the resulting LDBA. This construction ensures that $F'$ is visited infinitely if and only if $F_1,\ldots,F_n$ are met infinitely often. Further, the resulting LDBA conforms to our definition. We use the tool \textsc{Rabinizer 4}~\cite{kvretinsky2018rabinizer} which provide an equivalent LDBA 
for a given LTL formula.

For example, the LDGBA in Fig.~\ref{fig:ldgba}, with two acceptance sets \mbox{ $F_{1} = \{(s,e,s'),(t,f,t')\}$} and \mbox{ $F_{2} = \{(u,g,u')\}$}, can be reduced to an equivalent LDBA, as shown in Fig.~\ref{fig:ldba}, with acceptance set $F' = \{(s_1,e,s_{2}'),(t_1,f,t_{2}')\}$, by the method described above. On taking an accepting transition from $F'$ (corresponding to $F_1$ in the original LDGBA) starting from $Q_{f1}$, we end up in $Q_{f2}$ and need to make the transition 
$(u_2,g,u_{1}')$ (corresponding to $F_2$ in the original LDGBA) to revert to $Q_{f1}$ and be able to take an accepting transition from $F'$ again. Thus, $F'$ is visited infinitely if and only if $F_1$ and $F_2$ are met infinitely often in the given LDGBA.

\begin{figure}[t]
    \centering
    \includegraphics[scale=0.2]{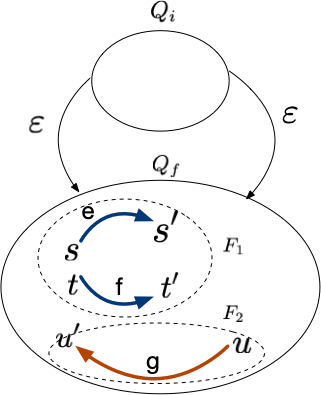}
    \caption{LDGBA with acceptance sets $F_{1}, F_{2}$.}
    \label{fig:ldgba}
\end{figure}

\begin{figure}[t]
    \centering
    \includegraphics[scale=0.2]{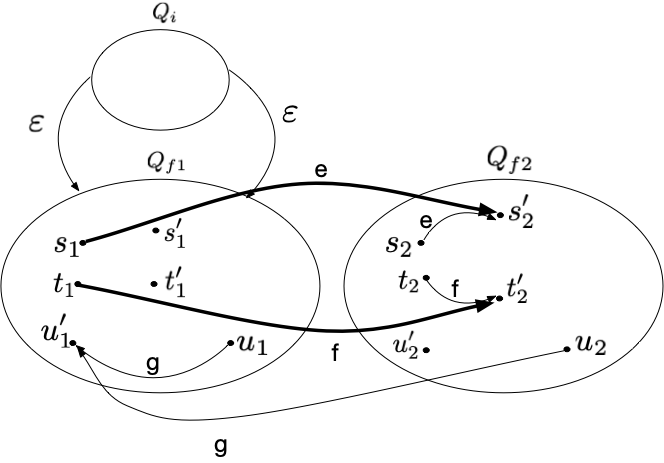}
    \caption{Equivalent LDBA with acceptance set $F'$.}
    \label{fig:ldba}
\end{figure}

\noindent\textbf{Occupancy Measures.}\label{sec:occu_measures} Occupancy measures \cite{altman1999constrained} allow formulating the problem of finding an optimal MDP policy as a linear program (LP) \cite{altman1999constrained,aaai2021}.

\noindent\underline{\textit{Discounted-Reward MDP.}} For a discounted MDP $\mathcal{M} = (S,A,P,s_0,\gamma,r)$, the occupancy measure $y^{\pi}: \mathbb{R}^{|S| \times |A|} \to \mathbb{R}_{\ge 0}$ of a stationary policy $\pi$ is defined as  $y_{sa}^{\pi} = \sum_{t=0}^{\infty} \gamma^t P_{\pi}(S_t=s,A_t=a),$
where $P_{\pi}(S_t=s,A_t=a)$ denotes the probability of being in state $s$ and taking action $a$ at time $t$ under policy $\pi$. In words, $y_{sa}^{\pi}$ is the (discounted) expected number of visits to state-action pair $(s,a)$ under policy $\pi$. The following linear inequalities, 
\begin{equation}\label{eq:ocdis}
    \begin{aligned}
    &y_{sa} \geq 0  \quad \forall s\in S \ \forall a \in A(s), \\
    &out(s) - in(s) =  \mathds{1}_{s_0}(s) \quad \forall s\in S,\\
    &out(s) = \sum_{a\in A(s)}y_{sa} \quad \quad \forall s\in S,\\
    &in(s) = \gamma \sum_{j\in S,a \in A(j)} y_{ja}P(j,a,s) \quad \forall s\in S, 
    \end{aligned}
\end{equation}
specify necessary and sufficient conditions for $y$
to be an occupancy measure for a discounted MDP $\mathcal{M}$~\cite{altman1999constrained}, where $out(s)$ and $in(s)$ denote the total \emph{outflow} from, and \emph{inflow} into, state $s$, respectively. The above inequalities can be interpreted as establishing the conservation of flow through each state. Moreover, the expected discounted reward under a stationary policy with occupancy measure $y$ can be expressed as $\sum_{s \in S} \sum_{a \in A(s)} r(s,a) y_{sa}$.

A corresponding stationary policy can be derived from an occupancy measure $y$ by setting $\pi(s,a)$ to be equal to $y_{sa} / \sum_{a \in A(s)} y_{sa}$ if $\sum_{a \in A(s)} y_{sa} > 0$, or arbitrary otherwise. Therefore, inequalities~\eqref{eq:ocdis} characterize the set of all stationary policies for $\mathcal{M}$.

\noindent\underline{\textit{Absorbing MDP.}}
An occupancy measure can also be defined for an absorbing MDP $\mathcal{M}_g = (S,A,P,s_0)$, that is, an MDP for which there exist a policy and a goal (absorbing) state $s_g \in S$ such that $s_g$ is reached almost surely. For an absorbing MDP $\mathcal{M}_g$ and policy $\pi$, an occupancy measure $x^{\pi}:  \mathbb{R}^{(|S|-1) \times |A|} \to \mathbb{R}_{\ge 0}$ can be defined~\cite{altman1999constrained} as $x_{sa}^{\pi} = \sum_{t=0}^{\infty} P_{\pi}(S_t=s,A_t=a), \quad \forall s \in S\setminus \{s_g\}$, interpreted as the expected number of visits to the state-action pair $(s,a)$ under policy $\pi$ before reaching the goal state. The measure is finite, as the goal state $s_g$ is reached almost surely.

Similarly to the case of discounted MDPs, the following set of linear inequalities
%
\allowdisplaybreaks \begin{align}\label{eq:ocssp}
    &x_{sa} \geq 0  \quad \forall s\in S\setminus \{s_g\} \ \forall a \in A(s), \nonumber \\
    &out(s) - in(s) =  \mathds{1}_{s_0}(s) \quad \forall s\in S\setminus \{s_g\}, \nonumber \\
    &in(s_g) = 1, \\
    &out(s) = \sum_{a\in A(s)}x_{sa} \quad \quad \forall s\in S\setminus \{s_g\},\nonumber\\
    &in(s) = \sum_{j\in S,a \in A(j)} x_{ja}P(j,a,s) \quad \forall s\in S,\nonumber 
    \end{align}
%
specify necessary and sufficient conditions for  $x$ to be an occupancy measure, where $out(s)$ and $in(s)$ are interpreted as in~\eqref{eq:ocdis} and $in(s_g) = 1$ ensures the almost sure reachability of the goal state. It is also possible to derive a corresponding stationary policy from an occupancy measure as described above. Therefore, the linear inequalities~\eqref{eq:ocssp} characterize the set of all stationary policies for $\mathcal{M}_{g}$ under which the goal state is reached almost surely.

\section{Problem Formulation and Solution Strategy} \label{sec:milp}

We aim to synthesize a policy $\pi$ for an MDP $\mathcal{M}$ that is optimal with respect to an expected discounted reward under the constraint that an LTL formula $\varphi$ is satisfied with probability 1. We make the following assumption which can be easily verified by using standard techniques from probabilistic model checking~\cite{KNP11}.
\begin{assumption}\label{asm}
For the given MDP $\mathcal{M}$ and LTL formula $\varphi$, there exists a policy such that the formula is satisfied with probability 1. 
\end{assumption}

We now introduce the methodology used to solve the above objective as illustrated in Fig.~\ref{fig:flowchart}. By representing the satisfaction of $\varphi$ with an LDBA $\mathcal{A}$ and leveraging a standard construction from LTL model checking of MDPs~\cite{sickert2016limit}, we generate a product MDP $\mathcal{M} \times \mathcal{A}$ as the composition of $\mathcal{M}$
and $\mathcal{A}$, such that almost sure satisfaction of a B\"{u}chi acceptance condition on the product MDP, shown to be possible by Assumption~\ref{asm}, implies the almost sure satisfaction of $\varphi$ on the original MDP. As further discussed in Section~\ref{sec:prodMDP}, it is sufficient to consider the set of stationary deterministic policies $\Pi_{d}^{\times}$ of the product MDP to find a policy satisfying the B\"{u}chi acceptance condition almost surely. Let $\Pi_{d}$ be the projection of $\Pi_{d}^{\times}$ onto the policy space of the original MDP $\mathcal{M}$. Assumption~\ref{asm} guarantees that a feasible policy exists in $\Pi_{d}$. We are then interested in the following problem:
\begin{equation} \label{eq:prob}
    \begin{aligned}
   \underset{\pi \in \Pi_d}{\text{ max }} \quad & {\mathcal R}_{\mathcal{M}}(\pi)\\ \mathrm{s.t.} \quad  & Pr_{\pi}^{\mathcal{M}}(\varphi) = 1.
\end{aligned}
\end{equation}

By building on a result from Hahn~et~al.~\cite{hahn2019omega}, we reduce the almost sure satisfaction of a B\"{u}chi acceptance condition on the product MDP to an almost sure reachability problem. We then introduce two occupancy measures to express both the reachability constraint and the discounted reward, and formulate a set of mixed integer linear constraints to guarantee that the occupancy measures are well defined and are both associated with the same deterministic policy. The solution of the resulting MILP provides the desired policy. In the following, we detail the above steps.  

\begin{figure}
    \centering
    \includegraphics[width=\linewidth]{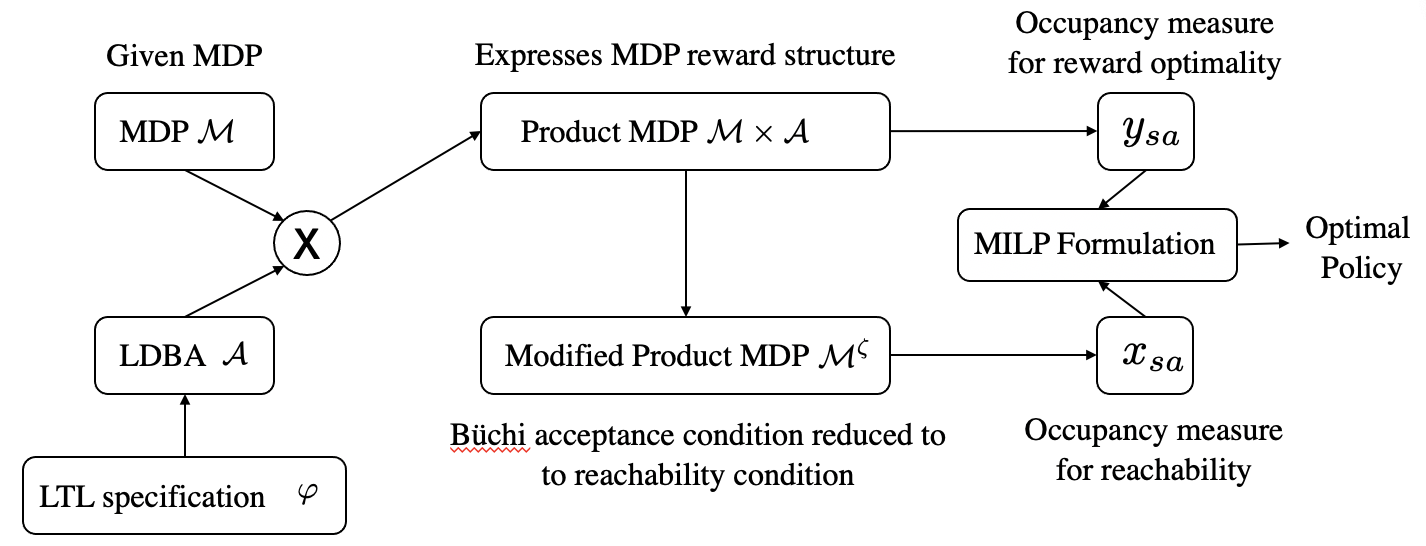}
    \caption{Flowchart describing the solution approach.}
    \label{fig:flowchart}
\end{figure}

\subsection{Construction of the Product MDP}
\label{sec:prodMDP}

Given an MDP $ \mathcal{M} = (S,A,P,s_0,\gamma,AP,L,r)$ and an LDBA $\mathcal{A}$ = $(Q, \Sigma, q_0, \delta, F)$ capturing the LTL formula $\varphi$, where $Q= Q_i \cup Q_f$, 
we follow the construction in~\cite{sickert2016limit} to define a product MDP ($\mathcal{M} \times \mathcal{A}$) $\mathcal{M} ^{\times} = (S^{\times},  A^{\times},P^{\times},s_{0}^{\times},\gamma,r^{\times}, F^{\times})$  which incorporates the transitions of $ \mathcal{M}$  and $ \mathcal{A}$, the reward function of $\mathcal{M}$ and the acceptance set of $\mathcal{A}$. The following running example illustrates the construction involved in our algorithm.

\begin{example}\label{example}

Consider the four state grid world MDP in Fig.~\ref{fig:4}. Each state associated with a location of the grid is labelled with the set of atomic propositions in  $\{l_0,l_1,m\}$, that are true in it. 
The transition diagram of the MDP is shown in Fig.~\ref{fig:5}. The initial state is III. The actions $\mathsf{rest}$ (denoted by red transition arrows) and $\mathsf{go}$ (denoted by green transition arrows) are available in each state. Action $\mathsf{rest}$ does not change the state of the MDP. With action $\mathsf{go}$, the agent moves horizontally, in the same row, with probability $p$ or vertically, in the same column, with probability $1 - p$.

\begin{figure}[h]
\centering
\begin{minipage}{.25\textwidth}
  \centering
  \includegraphics[width=.5\linewidth]{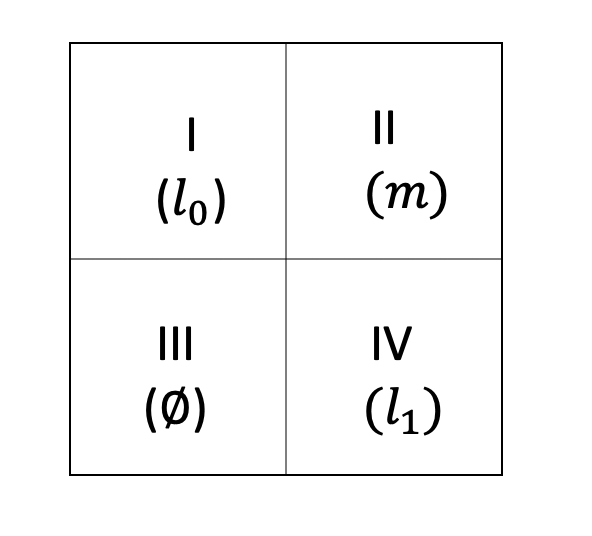}
  \caption{Grid world example.}
  \label{fig:4}
\end{minipage}%
\begin{minipage}{.25\textwidth}
  \centering
  \includegraphics[width=.7\linewidth]{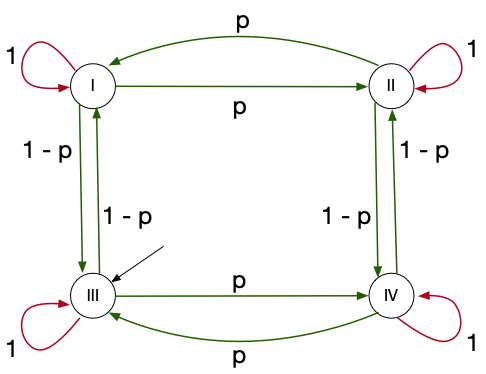}
  \caption{MDP transition diagram.}
  \label{fig:5}
\end{minipage}
\end{figure}

 We wish to satisfy the LTL formula $\varphi :=  (\textbf{FG} l_0 \vee \textbf{FG} l_1) \wedge \textbf{G} \neg m$, requiring to reach one of the safe cells (where $l_0$ or $l_1$ is true) and stay there forever, while avoiding the unsafe cell, where $m$ holds. An equivalent LDBA for this LTL formula is shown in Fig.~\ref{fig:6}. The transitions are labelled by propositional formulae which are equivalent to the Boolean assignments of the atomic propositions $\{l_0,l_1,m\}$ which satisfy the propositional formulae. The initial state is $0$ and the accepting transitions are marked in bold.
 \begin{figure}[h]
    \centering
    \includegraphics[scale=0.2]{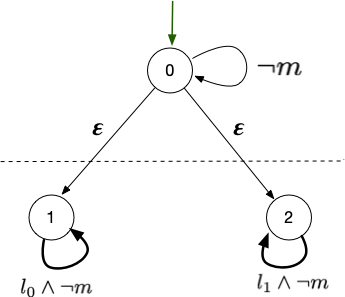}
    \caption{LDBA for $\varphi :=  (\textbf{FG} l_0 \vee \textbf{FG} l_1) \wedge \textbf{G} \neg m $ }
    \label{fig:6}
\end{figure}
 The only non-deterministic transitions are the $\epsilon$ transitions from $Q_i =  \{0\}$ to $Q_f = \{1,2\}$, where $Q_i \cap Q_f = \emptyset$ and $Q_i \cup Q_f = Q = \{0,1,2\}$. Further, $F \subseteq Q_f \times \Sigma \times Q_f$ holds, and there are no transitions from $Q_f$ to $Q_i$. As denoted by the dotted line, the states can be partitioned into $Q_i$ and $Q_f$.
\end{example}

In the product MDP $\mathcal{M} ^{\times}$, $S^{\times} = S \times Q$ is the set of states, $A^{\times} = A \cup E $ is the action set, where $E = \{ \epsilon_q | q \in Q_f \}$ is the set of actions which simulate $\epsilon$-transitions to states in $Q_f$, and $s_{0}^{\times} = (s_0,q_0)$ is the initial state. We then define the transition function $P^{\times}((s,q),a,(s',q'))$ as follows
\begin{equation} \label{eq:prodtrans}
\begin{aligned}
& \begin{cases} P(s,a,s'), &\mbox{if } q' = \delta(q,L(s)) \text{ and } a \not \in E, \\
1, &\mbox{if } s = s',\quad a = \epsilon_{q'}, \quad q' \in \delta(q,\epsilon), \\
0, & \text{otherwise.}
\end{cases}\\
\end{aligned}    
\end{equation}
and the reward function as
\begin{equation} \label{eq:prodrew}
\begin{aligned}
r^{\times}((s,q),a) & = \begin{cases} r(s,a), &\mbox{if }  a \not \in E , q \in Q_i, \\
r(s,a) / \gamma, &\mbox{if }  a \not \in E , q \in Q_f, \\
0, & \text{otherwise.}
\end{cases}\\
\end{aligned}    
\end{equation}
%
Finally, the set of accepting transition of $\mathcal{M}^{\times}$ includes the non-zero probability transitions of $\mathcal{M}^{\times}$ capturing the accepting transitions of  $\mathcal{A}$, i.e., $F^{\times} = \{ ((s_i,q_i),a,(s_j,q_j)) | (q_i,L(s_i),q_j) \in F , P(s_i,a,s_j) > 0  \}$. 

\begin{example}\label{example_product}
Fig.~\ref{fig:7} shows the product MDP $\mathcal{M}^{\times}$ between $\mathcal{M}$ and $\mathcal{A}$ in Example~\ref{example}.

\begin{figure}[h]
    \centering
    \includegraphics[scale=0.25]{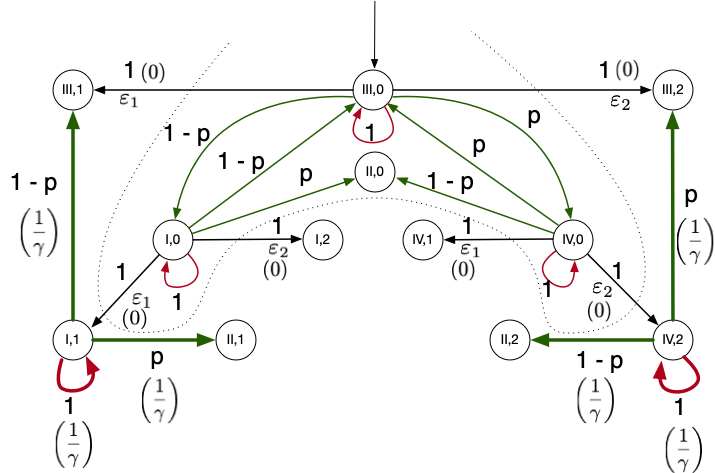}
    \caption{Product MDP $\mathcal{M}^{\times}$.}
    \label{fig:7}
\end{figure}

The $\mathsf{rest}$, $\mathsf{go}$, and $\epsilon$ actions are denoted by red, green, and black arrows, respectively. The accepting transitions $F^{\times}$ are in bold.  
The rewards associated with the discounted reward objective 
are obtained as follows. Epsilon actions $\epsilon_1$ and $\epsilon_2$ are associated with $0$ reward (in brackets). The rewards associated with transitions from states $(s,q)$, $q\in Q_f$, are multiplied by $\frac{1}{\gamma}$ to compensate for the artificial increase in time step and additional discounting due to the $\epsilon$-transition, as further detailed below. We observe that, differently from previous work~\cite{hahn2019omega,bozkurt2019control,hasanbeig2019certified}, the expected discounted reward in our setting is associated with a user-defined objective function and is unrelated to the LTL reachability objective.
\end{example}

A run $\xi^{\times}$ of $\mathcal{M}^{\times}$ is said to satisfy the B\"{u}chi acceptance condition $\varphi_{F^{\times}}$ if $\inf(\xi^{\times}) \cap F^{\times} \neq \emptyset$, i.e., there exists a transition belonging to $F^{\times}$ that occurs infinitely often in $\xi^{\times}$. The probability that a run of  $\mathcal{M}^{\times}$ starting from the initial state $s_0^{\times}$ under policy $\pi^{\times}$ satisfies the B\"{u}chi acceptance condition $\varphi_{F^{\times}}$ is denoted by $Pr_{\pi^{\times}}^{\mathcal{M}^{\times}}(\varphi_{F^{\times}})$. 

The transition function in~\eqref{eq:prodtrans} represents the effect of $\epsilon$-transitions in $\mathcal{A}$, which are the only non-deterministic transitions,  in terms of actions $\epsilon_{q^{'}}$ that cause the current state $(s,q)$ to transition to a state $(s',q')$, where  
$q'$ 
is the destination of the corresponding $\epsilon$-transition in $\mathcal{A}$, while $s'=s$, 
meaning that no change in the corresponding  state of $\mathcal{M}$ takes place. 
Only when an action $a \not\in E$ is taken, the corresponding state $s$ of $\mathcal{M}$ is updated based on the transition  probability $P(s,a,s')$, while $q$ is updated as determined by the corresponding transition function $\delta$ of $\mathcal{A}$ and the label $L(s)$ associated with $s$. 

We observe that a run $\xi^{\times}$ of $\mathcal{M}^{\times}$ can contain at most one $\epsilon$-transition, corresponding to a jump from $Q_i$ to $Q_f$ in $\mathcal{A}$. A run 
$(s_0,q_0)a_0(s_1,q_1)a_1 \ldots (s_l,q_l) \epsilon_{q_{l+1}}(s_{l+1},q_{l+1})a_{l+1} \ldots$, with $s_{l+1} = s_l$, corresponds to the run $s_0,a_0,s_1,a_1,\ldots,s_{l-1},a_{l-1},s_{l},a_{l+1},s_{l+2},a_{l+2},\ldots$ of $\mathcal{M}$, where we observe that 
an $\epsilon$-action in $\mathcal{M}^{\times}$ does not reflect into a change of state for $\mathcal{M}$. 
The additional division by $\gamma$ introduced in~\eqref{eq:prodrew} compensates for the  discounting caused by the $\epsilon$-transition in $\mathcal{M}^{\times}$ and the resulting additional time-step, as stated by the lemma below.
\vspace{3pt}
%

\begin{lemma}\label{lem:rwd}
Given an MDP $\mathcal{M}= (S,A,P,s_0,\gamma,AP,L,r)$, an LTL formula $\varphi$, and an LDBA $\mathcal{A} = (Q, \Sigma, q_0, \delta, F)$ equivalent to $\varphi$, let $\mathcal{M}^{\times} = \mathcal{M} \times \mathcal{A}$. The total discounted reward of a run $\xi^{\times}$ 
of $\mathcal{M}^{\times}$ is equal to the total discounted reward of $\xi$, obtained by  projecting $\xi^{\times}$ to the state and action spaces of $\mathcal{M}$.
\end{lemma}

\begin{proof} 
We consider two cases, based on whether $\xi^{\times}$ contains an $\epsilon$-transition or not. Let us first assume that $\xi^{\times}$ does not contain $\epsilon$-transitions. Then, by~\eqref{eq:prodrew} and the fact that $q_n \in Q_i $ for all $n \geq 0$, we obtain 
\begin{equation*}
    \sum_{n=0}^{\infty} { \gamma }^{n} r^{\times}((s_n,q_n),a_n) = \sum_{n=0}^{\infty} { \gamma }^{n} r(s_n,a_n). 
\end{equation*}
This implies that $G_{\mathcal{M}^{\times}}(\xi^{\times}) = G_{\mathcal{M}}(\xi)$, i.e., the total discounted reward of $\xi^{\times}$ is equal to that of its projection $\xi^{\times}$ to the state and action spaces of $\mathcal{M}$.

Assume now that $\xi^{\times}$ contains an $\epsilon$-transition, that is, given  $\xi^{\times} = (s_0,q_0)a_0 \ldots (s_l,q_l) a_l (s_{l+1},q_{l+1}) \ldots $, we have $a_l = \epsilon_{q_{l+1}}$ and $s_{l+1}=s_l$. Then, by~\eqref{eq:prodrew}, we obtain
\begin{align*}
& \sum_{n=0}^{\infty} { \gamma }^{n}  r^{\times}((s_n,q_n),a_n) = \sum_{n=0}^{l-1} { \gamma }^{n} r^{\times}((s_n,q_n),a_n)) + \\
& \qquad +{\gamma}^{l}r^{\times}((s_l,q_l),\epsilon_{q_{l+1}})  + \sum_{n=l+1}^{\infty}{\gamma}^{n} r^{\times}((s_n,q_n),a_n),
\end{align*}
where $q_n \in Q_i $ for $n \leq l$ and $q_n \in Q_f $ for $n > l$. Further, by observing that $s_{l+1} = s_l$, and $r^{\times}((s_l,q_l),\epsilon_{q_{l+1}}) = 0$, we obtain
\begin{equation*}
 \sum_{n=0}^{\infty} { \gamma }^{n} r^{\times}((s_n,q_n),a_n) = \sum_{n=0}^{l-1} { \gamma }^{n} r(s_n,a_n)  + 
 \end{equation*}
 \begin{equation*}
(1/\gamma)(\sum_{n=l+1}^{\infty} { \gamma }^{n} r(s_n,a_n))
 \end{equation*}
which is the total discounted reward associated with the run $\xi: (s_0,a_0,\ldots,s_{l-1},a_{l-1},s_{l},a_{l+1},s_{l+2},a_{l+2},\ldots)$. Again, we conclude $G_{\mathcal{M}^{\times}}(\xi^{\times}) = G_{\mathcal{M}}(\xi)$, which is what we wanted to prove. 
\qed 
\end{proof}

By Theorem~$3$ in~\cite{sickert2016limit}, it is sufficient to focus on stationary policies of $\mathcal{M}^{\times}$ to maximize the probability of satisfaction of the LTL formula $\varphi$. 
We restate this result in Lemma~\ref{lem:buch} below.


\begin{lemma}\label{lem:buch}
For $\mathcal{M}$ and $\mathcal{M}^{\times}$ as defined in Lemma~\ref{lem:rwd}, any stationary policy $\pi^{\times}_{*}$ of $\mathcal{M}^{\times}$ which maximizes the probability of satisfying the   B\"{u}chi acceptance condition $Pr_{\pi^{\times}}^{\mathcal{M}^{\times}}(\varphi_{F^{\times}})$ in $\mathcal{M}^{\times}$  induces a finite memory policy $\pi_{*}$ on  $\mathcal{M}$ such that $\pi_{*}$ maximizes the probability of LTL satisfaction $Pr_{\pi}^{\mathcal{M}}(\varphi)$ and $Pr_{\pi^{\times}_{*}}^{\mathcal{M}^{\times}}(\varphi_{F^{\times}}) = Pr_{\pi_{*}}^{\mathcal{M}}(\varphi)$ holds.
\end{lemma}

A stationary policy $\pi^{\times}$ for $\mathcal{M}^{\times}$ can be mapped to  $\pi$ for $\mathcal{M}$ as follows. $\mathcal{M}$ and $\mathcal{A}$ start from states $s_0$ and $q_0$, respectively. Whenever $\mathcal{M}$ transitions from $s$ to $s'$ on action $a \in A(s)$, $\mathcal{A}$ updates its state to $\delta(q,L(s))$. For states $s$ of $\mathcal{M}$ and $q$ of $\mathcal{A}$, action $a \in A^{\times}$ is selected with probability $\pi^{\times}((s,q),a)$. If the selected action is an $\epsilon$-action $\epsilon_{q'}$, then the state of $\mathcal{A}$ is updated to $q'$ while  $\mathcal{M}$ keeps the same state and selects the next action $a \in A$ with probability $\pi^{\times}((s,q'),a)$. Else, both $\mathcal{M}$ and $\mathcal{A}$ progress based on their transition function. Therefore, $\pi$ prescribes an action given the current state of 
$\mathcal{M}$ and $\mathcal{A}$ based on  $\pi^{\times}$.




By Lemma~\ref{lem:buch} and Assumption~\ref{asm}, there exist stationary policies for $\mathcal{M}^{\times}$ such that $Pr_{\pi^{\times}}^{\mathcal{M}^{\times}}(\varphi_{F^{\times}}) = Pr_{\pi}^{\mathcal{M}}(\varphi) = 1$. Moreover, in Section~\ref{sec:reachred}, $Pr_{\pi^{\times}}^{\mathcal{M}^{\times}}(\varphi_{F^{\times}}) = 1$ is shown to be equivalent to an almost sure reachability, for which optimal \emph{deterministic} stationary policies are guaranteed to exist~\cite{Puterman:1994:MDP:528623}. We can then look for a deterministic stationary policy $\pi^{\times}$ such that $Pr_{\pi^{\times}}^{\mathcal{M}^{\times}}(\varphi_{F^{\times}}) = 1$ and ${\mathcal R}_{\mathcal{M}^{\times}}(\pi^{\times})$ is maximized.  
Finally, by Lemma~\ref{lem:rwd} and the construction of the product MDP, we have that ${\mathcal R}_{\mathcal{M}}(\pi) = {\mathcal R}_{\mathcal{M}^{\times}}(\pi^{\times})$ holds, and Problem~\eqref{eq:prob} reduces to the following problem: 
\begin{equation} \label{eq:prob2}
    \begin{aligned}
   \underset{\pi^{\times} \in \Pi_{d}^{\times}}{\text{ max }} \quad & {\mathcal R}_{\mathcal{M}^{\times}}(\pi^{\times})\\ \mathrm{s.t.} \quad  & Pr_{\pi^{\times}}^{\mathcal{M}^{\times}}(\varphi_{F^{\times}}) = 1.
\end{aligned}
\end{equation}

\subsection{Reachability Reduction of B\"{u}chi Acceptance} \label{sec:reachred}

Recent work~\cite{hahn2019omega} has shown that a modified product MDP $\mathcal{M}^{\zeta}$ can be generated from  $\mathcal{M}^{\times}$ such that almost sure satisfaction of the B\"{u}chi acceptance condition is equivalent to almost sure reachability of a new absorbing state in $\mathcal{M}^{\zeta}$.

The modified product MDP $\mathcal{M}^{\zeta}$ is defined as  $(S^{\zeta}, A^{\zeta},P^{\zeta},s_0^{\zeta})$, where $S^{\zeta} = S^{\times} \cup \{g\} $, $A^{\zeta} = A^{\times}$ and $s_0^{\zeta} = s_0^{\times}$. $\mathcal{M}^{\zeta}$ is constructed from $\mathcal{M}^{\times}$ by adding a new absorbing state $g$ to $S^{\times}$. Given $\zeta \in (0,1)$, $P^{\zeta}$ is obtained from $P^{\times}$ as follows. For each accepting transition $((s_i,q_i),a,(s_j,q_j)) \in F^{\times} $, $g$ becomes the destination of the transition with probability $1 - \zeta$, while the  probabilities of reaching all the other possible destinations on taking action $a$ in $(s_i,q_i)$ are multiplied by $\zeta$.

\begin{example}
Fig.\ref{fig:8} shows the MDP $\mathcal{M}^{\zeta}$ obtained by modifying the product MDP $\mathcal{M}^{\times}$ in Fig.~\ref{fig:7} to add a new absorbing state $g$ and modified accepting transitions as described above. The new transitions in $\mathcal{M}^{\zeta}$ are marked by dot-dash lines. $\mathcal{M}^{\zeta}$  does not have an associated reward function and is only used to express the almost satisfaction of the given LTL formula. On the other hand, $\mathcal{M}^{\times}$ in Fig.~\ref{fig:7} is used to express the discounted-reward objective. 

\begin{figure}[h]
    \centering
    \includegraphics[scale=0.25]{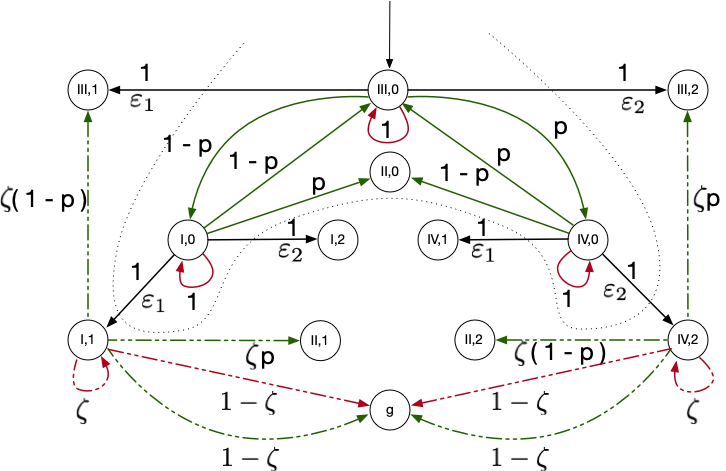}
    \caption{Modified product MDP $\mathcal{M}^{\zeta}$.}
    \label{fig:8}
\end{figure}
\end{example}

We observe that $\mathcal{M}^{\zeta}$ and  $\mathcal{M}^{\times}$ have a common set of states $S^{\times}$ except for the absorbing state $g$ in $S^{\zeta}$. Therefore, a single stationary policy can be defined for both $\mathcal{M}^{\zeta}$ and  $\mathcal{M}^{\times}$ over $S^{\times}$. Let $\rho_{\pi^{\times}}^{\zeta}(t)$ denote the probability of reaching $g$ in $\mathcal{M}^{\zeta}$ by following policy $\pi^{\times}$. The following lemmas relate the probability of reaching $g$ in $\mathcal{M}^{\zeta}$ with the probability of satisfying the LTL formula $\varphi$, i.e., of satisfying the associated B\"{u}chi acceptance condition in $\mathcal{M}^{\times}$. Lemma~\ref{lem:ineq} below follows from Lemma 1~\cite{hahn2019omega}.  

 \begin{lemma}\label{lem:ineq}
 For $\mathcal{M}$ and $\mathcal{M}^{\times}$ as defined in Lemma~\ref{lem:rwd}, let $\mathcal{M}^{\zeta} = (S^{\zeta}, A^{\zeta},P^{\zeta},s_0^{\zeta})$ be the modified absorbing product MDP.  For $\zeta \in (0,1)$ and a stationary policy $\pi^{\times}$ in $\mathcal{M}^{\zeta}$, we obtain $\rho_{\pi^{\times}}^{\zeta}(t) \geq Pr_{\pi^{\times}}^{\mathcal{M}^{\times}}(\varphi_{F^{\times}})$. Moreover, $\rho_{\pi^{\times}}^{\zeta}(t) = 1 $ implies $Pr_{\pi^{\times}}^{\mathcal{M}^{\times}}(\varphi_{F^{\times}}) = 1$. 
  \end{lemma}

In the case of almost sure satisfaction, Lemma~\ref{lem:ineq} reduces to the following result:
 \begin{lemma}\label{lem:reach}
 For $\mathcal{M}$ and $\mathcal{M}^{\times}$ as defined in Lemma~\ref{lem:rwd}, let $\mathcal{M}^{\zeta} = (S^{\zeta}, A^{\zeta},P^{\zeta},s_0^{\zeta})$ be the modified absorbing product MDP.
  For $\zeta \in (0,1)$ and a stationary policy $\pi^{\times}$ in $\mathcal{M}^{\zeta}$, $\rho_{\pi^{\times}}^{\zeta}(t) = 1 $ if and only if  $Pr_{\pi^{\times}}^{\mathcal{M}^{\times}}(\varphi_{F^{\times}}) = 1$.
 \end{lemma}

\begin{proof}
 If $\pi^{\times}$ is such that $\rho_{\pi^{\times}}^{\zeta}(t) = 1 $, then  $Pr_{\pi^{\times}}^{\mathcal{M}^{\times}}(\varphi_{F^{\times}}) = 1$ by Lemma \ref{lem:ineq}. Also, it is obvious from Lemma \ref{lem:ineq} that $Pr_{\pi^{\times}}^{\mathcal{M}^{\times}}(\varphi_{F^{\times}}) = 1$ implies $\rho_{\pi^{\times}}^{\zeta}(t) = 1 $. 
 \qed
\end{proof}

We conclude that the set of stationary deterministic policies of $\mathcal{M}^{\times}$ which satisfy the B\"{u}chi acceptance condition almost surely is equal to the set of stationary deterministic policies of $\mathcal{M}^{\zeta}$ for which the absorbing state $g$ is reached almost surely.
 
\subsection{Occupancy Measures}\label{sec:OM}

From Lemmas~\ref{lem:buch} and~\ref{lem:reach}, almost sure satisfaction of an LTL constraint on $\mathcal{M}$ can be expressed as an almost sure  reachability constraint on $\mathcal{M}^{\zeta}$. 
However, this reachability constraint can also be expressed in terms of the occupancy measure $x_{sa}$ defined on the absorbing MDP $\mathcal{M}^{\zeta}$ as discussed Section~\ref{sec:occu_measures}:
\begin{equation}\label{occx}
    \begin{aligned}
     &x_{sa} \geq 0  \quad \forall s\in S^{\times} \ \forall a \in A^{\times}(s),\\
    &out^{x}(s) - in^{x}(s) =  \mathds{1}_{s_0^{\times}}(s) \quad \forall s\in S^{\times},\\
    &in^{x}(g) = 1 \quad \quad \quad g \in S^{\zeta},\\
    &out^{x}(s) = \sum_{a\in A^{\times}(s)}x_{sa} \quad \quad \forall s\in S^{\times}, \\
    &in^{x}(s) = \sum_{j\in S^{\times},a \in A^{\times}(j)} x_{ja}P^{\zeta}(j,a,s) \quad \forall s\in S^{\zeta}.
    \end{aligned}
\end{equation}

Similarly, the discounted reward optimization problem on $\mathcal{M}^{\times}$ can be expressed in terms of an occupancy measure $y_{sa}$ $\mathcal{M}^{\times}$ as follows:
\begin{equation}\label{occy}
    \begin{aligned}
    &y_{sa} \geq 0  \quad \forall s\in S^{\times} \ \forall a \in A^{\times}(s),\\
    &out^{y}(s) - in^{y}(s) =  \mathds{1}_{ s_0^{\times}}(s) \quad \forall s\in S^{\times},\\
    &out^{y}(s) = \sum_{a\in A^{\times}(s)}y_{sa}, \quad \quad \forall s\in S^{\times},\\
    &in^{y}(s) = \gamma\sum_{j\in S^{\times},a \in A^{\times}(j)} y_{ja}P^{\times}(j,a,s) \quad \forall s\in S^{\times}.
    \end{aligned}
\end{equation}
while the expected discounted reward under policy $\pi^{\times}$ with occupancy measure $y_{sa}$ is ${\mathcal R}_{\mathcal{M}^{\times}}({\pi}^{\times}) =  \sum_{s \in S^{\times}, a \in A^{\times}(s)}  r^{\times}(s,a) y_{sa}.$

Finally, we need to ensure that both occupancy measures originate from the same underlying policy $\pi^\times$, which maximizes the reward and satisfies the LTL specification. We then require that $ {x_{sa}}/{\sum_{a^{'} \in A^{\times}(s)} x_{sa^{'}}} = {y_{sa}}/{\sum_{a^{'} \in A^{\times}(s)} y_{sa^{'}}},$
for all $a \in A^{\times}(s)$ and $s\in S^{\times}$. We achieve this equality by introducing a set of binary variables and constraints to restrict the search to the space of deterministic policies as detailed in the next section.

\subsection{Mixed Integer Linear Program Formulation}\label{sec:delta}

For all the states shared by $\mathcal{M}^{\times}$ and $\mathcal{M}^{\zeta}$ and the associated actions, i.e., for all $s \in S^{\times}$ and $a \in A^{\times}(s)$, we introduce binary variables $\Delta_{sa} \in \{0,1 \}$ that evaluate to $1$ if and only if action $a$ is selected in visited state $s$. 
%
%
By using these binary variables, we can require that a policy is deterministic via the following constraint:
\begin{equation}\label{deter}
    \sum_{a} \Delta_{sa} \leq 1 \quad \forall s \in S^{\times},
\end{equation}
meaning that at most one action $a$ can be selected in state $s$. 
We can then require the equivalence between the occupancy measures defined in Section~\ref{sec:OM} using the following logical constraints:
%
%
\begin{subequations}\label{g_zero}
 \begin{align}
& (\Delta_{sa} = 0) \to (x_{sa} = 0) \quad \forall s \in S^{\times}, \ 
\forall a \in A^{\times}, \label{gxzero}\\
& (\Delta_{sa} = 0) \to (y_{sa} = 0) \quad  \forall s \in S^{\times}, \ \forall a \in A^{\times}, \label{gyzero} 
\end{align}  
\end{subequations}
where $\to$ denotes the logical implication. These constraints can be converted into mixed integer linear constraints using standard techniques~\cite{Winston04}. A similar idea was previously used to generate deterministic policies for constrained MDPs~\cite{dolgov2005stationary}.

The policies generated by the occupancy measures are deterministic under constraints \eqref{deter}-\eqref{g_zero} as stated by the following lemmas.

\begin{lemma}\label{deterproof1}
Given a discounted MDP $\mathcal{M}^{\times}$, a stationary policy $\pi^\times$, the associated occupancy measure $y_{sa}$, binary variables $\Delta_{sa}$, $\forall s \in S^{\times}, \ \forall a \in A^{\times}$, if $y$ and $\Delta$ satisfy equations~\eqref{deter}-\eqref{g_zero}, then for all states $s \in S^{\times}$ visited with non-zero probability, policy $\pi^{\times}$ is deterministic and, further, $y_{sa} > 0$ if and only if $\Delta_{sa} = 1$. 
\end{lemma}

\begin{proof}
Consider a state $s^* \in S^{\times}$ which is visited with non-zero probability under policy $\pi^{\times}$. We obtain the following chain of implications: 
\begin{align*}
    & \sum_{a \in A^{\times}(s^*)} y_{s^{*}a} > 0 \implies \exists a^* \in A^{\times}(s^*) \ \text{s.t.} \ y_{s^{*}a^{*}} > 0 \\
    & \implies \Delta_{s^{*}a^{*}} = 1 \text{ by \eqref{gyzero}} \\
    & \implies \Delta_{s^{*}a} = 0 \quad  \forall a \neq a^{*} \text{ by \eqref{deter}} \\
    & \implies y_{s^{*}a} = 0 \quad  \forall a \neq a^{*} \text{ by \eqref{gyzero}}\\
    & \implies \pi^{\times}(s^{*},a^{*}) = 1 \quad \text{and} \quad \pi^{\times}(s^{*},a) = 0 \quad  \forall a \neq a^{*}  \\
    & \implies \pi^{\times} \text{ is deterministic.}
\end{align*}
Therefore, if $s$ is visited with non-zero probability, $y_{sa} > 0$ if and only if $\Delta_{sa} = 1$. 
\qed
\end{proof}

Similarly, the following lemma can be stated for the absorbing MDP $\mathcal{M}^{\zeta}$.
 
\begin{lemma}\label{deterproof2}
Given an absorbing MDP $\mathcal{M}^{\zeta}$, a stationary policy $\pi^{\times}$, the associated occupancy measure $x_{sa}$, binary variables $\Delta_{sa}$, $\forall s \in S^{\times}, \ \forall a \in A^{\times}$, if $x$ and $\Delta$ satisfy equations~\eqref{deter}-\eqref{g_zero}, then for all states $s \in S^{\zeta} \setminus \{g\}$ visited with non-zero probability, policy $\pi^{\times}$ is deterministic and, further, $x_{sa} > 0$ if and only if $\Delta_{sa} = 1$.
\end{lemma}

The following lemma shows that under constraints~\eqref{deter}-\eqref{g_zero}, the set of reachable states in $S^{\times}$ under policies given by $x$ and $y$ are the same and have the same deterministic action choices on those states.

\begin{lemma}\label{reach}
Given an MDP $\mathcal{M}^{\times}$, a stationary policy $\pi^y$, the associated occupancy measure $y_{sa}$, the absorbing MDP $\mathcal{M}^{\zeta}$ generated from $\mathcal{M}^{\times}$, a stationary policy $\pi^x$, the associated occupancy measure $x_{sa}$, binary variables $\Delta_{sa}$, $\forall s \in S^{\times}, \ \forall a \in A^{\times}$, if $x, y$ and $\Delta$ satisfy equations~\eqref{deter}-\eqref{g_zero}, then the set of reachable states in $S^{\times}$ under $\pi_x$ and $\pi_y$ and the deterministic action choices on those states are the same. 
\end{lemma}

\begin{proof}
$\mathcal{M}^{\times}$ and $\mathcal{M}^{\zeta}$ have the same initial state $(s_0,q_0)$, which is clearly visited for both the MDPs. Therefore, from the proof of Lemma \ref{deterproof1},  $ \exists a^{*} \text{ such that } \Delta_{(s_0,q_0),a^{*}} = 1$ and  $\Delta_{(s_0,q_0),a} = 0 \ \forall a \neq a^{*}$. From Lemma~\ref{deterproof1} and Lemma~\ref{deterproof2}, we also have $x_{(s_0,q_0),a^*} > 0, y_{(s_0,q_0),a^*} > 0$ as well as  
$x_{(s_0,q_0),a} = 0, y_{(s_0,q_0),a} = 0, \forall a \neq a^{*}$. Therefore, both the occupancy measures provide the same action $a^*$ in the initial state.

The set of states reachable in $S^{\times}$ after one step, that is, after selecting $a^*$ in the initial state, is the same for both $\mathcal{M}^{\times}$ and $\mathcal{M}^{\zeta}$ by the construction of the transition probabilities $P^{\zeta}$ from $P^{\times}$. In fact, the transition probabilities between states of $S^{\times}$ are scaled by non-zero constants, which does not change their reachability. Similarly to the argument above, both occupancy measures propose the same actions in all the states reachable in $S^{\times}$ after one step. We can then proceed by induction to conclude that the set of reachable states in $S^{\times}$ for  $\mathcal{M}^{\times}$ and $\mathcal{M}^{\zeta}$ are the same for $n$ steps, $n \geq 0$, and both the occupancy measures determine the same actions in all the reachable states in $S^{\times}$.  \qed
\end{proof}

By putting together constraints~\eqref{occx}-\eqref{g_zero} in terms of the binary variables, occupancy measures, and the expected discounted reward objective, we obtain the following MILP problem whose solution provides the desired policy:

\allowdisplaybreaks\begin{align}\label{milp}
 &  \underset{x,y,\Delta\in \{0,1\}}{\text{ max }} \quad   \sum_{s \in S^{\times},a \in A^{\times}(s)}  r^{\times}(s,a) y_{sa},   \\ 
   \mathrm{s.t.} \quad  &x_{sa} \geq 0, \quad  y_{sa} \geq 0 \quad \forall s\in S^{\times} \ \forall a \in A^{\times}(s), \nonumber \\
   &out^{y}(s) - in^{y}(s) =  \mathds{1}_{ s_0^{\times}}(s) \quad \forall s\in S^{\times},  \nonumber  \\ 
   &out^{x}(s) - in^{x}(s) =  \mathds{1}_{s_0^{\times}}(s) \quad \forall s\in S^{\times}, \nonumber  \\
   &in^{x}(g) = 1 \quad \quad \quad g \in S^{\zeta}, \nonumber \\
   & \sum_{a \in A^{\times}(s)} \Delta_{sa} \leq 1 \quad \forall s \in S^{\times}, \nonumber  \\
   & (\Delta_{sa} = 0) \to (x_{sa} = 0) \quad \forall s \in S^{\times} \ \forall a \in A^{\times}(s), \nonumber \\
 & (\Delta_{sa} = 0) \to (y_{sa} = 0) \quad \forall s \in S^{\times} \ \forall a \in A^{\times}(s),  \nonumber \\
 &out^{y}(s) = \sum_{a\in A^{\times}(s)}y_{sa}, \quad \quad \forall s\in S^{\times},\nonumber \\
 &in^{y}(s) = \gamma\sum_{j\in S^{\times},a \in A^{\times}(j)} y_{ja}P^{\times}(j,a,s) \quad \forall s\in S^{\times},\nonumber \\
 &out^{x}(s) = \sum_{a\in A^{\times}(s)}x_{sa} \quad \quad \forall s\in S^{\times}, \nonumber \\
 &in^{x}(s) = \sum_{j\in S^{\times},a \in A^{\times}(j)} x_{ja}P^{\zeta}(j,a,s) \quad \forall s\in S^{\zeta}. \nonumber
\end{align}

The main result of this section can then be  summarized by the following theorem.

\begin{theorem}
Under Assumption~\ref{asm}, Problem~\eqref{milp} is feasible and provides an optimal deterministic policy for the formulated Problem~\eqref{eq:prob2}. 
Specifically, if $\Delta_{sa} = 1$, then the optimal deterministic policy requires to select action $a$ on visiting state $s$. 
\end{theorem}

Problem~\eqref{milp} is feasible since a deterministic policy which reaches the absorbing state $g$ with probability 1 exists due to Assumption~\ref{asm} and  Lemmas~\ref{lem:buch} and~\ref{lem:reach}. This deterministic policy on $\mathcal{M}^{\times}$ translates into a a finite-memory policy on $\mathcal{M}$ as discussed in Section~\ref{sec:prodMDP}. 

\subsection{Incorporating Additional Expected Discounted-Reward Constraints}

In certain applications, it may be necessary to express additional constraints that can be captured in terms of lower bounds to  expected total discounted returns $d_i$, $ i = 1,\ldots, n$, with respect to reward functions $r_i$ and discount factors $\gamma_i$ which are different from the primary function $r$ and discount factor $\gamma$, respectively. These constraints of the form   
$${\mathbb E}_{\pi} \left[\sum_{t=0}^{\infty} { \gamma_i }^{t} r_{i}(s_t,a_t) \right] > d_i, \quad i = 1,\ldots,n$$ 
can be seamlessly incorporated within our formulation, by defining additional occupancy measures  $y^{i}$, associated with $r_i,\gamma_i$, $ i = 1,\ldots,n$, as 
$$\sum_{s \in S^{\times},a \in A^{\times}(s)}  r^{\times}(s,a) y^{i}_{sa}  > d_i \quad  i = 1 \ldots n $$
along with the necessary and sufficient conditions for $y^{i}$ to be an occupancy measure as in~\eqref{occy}. Finally, we ensure that the occupancy measures originate from the same underlying policy by the following additional constraints, similar to the ones in Section~\ref{sec:delta}:
$$  (\Delta_{sa} = 0) \to (y^{i}_{sa} = 0) \quad \forall s \in S^{\times} \ \forall a \in A^{\times}(s).$$

\section{Experimental Results} \label{sec:experiment}

We implemented the MILP-based framework in \textsc{Python}, using \textsc{Rabinizer 4}~\cite{kvretinsky2018rabinizer} to generate an equivalent LDBA from an LTL formula and \textsc{Gurobi} for solving the MILPs. We evaluate the framework on two case studies involving motion planning of a mobile robot. The experiments are run on a $1.4$-GHz Core i5 processor with $16$-GB memory.

\subsection{Safe Motion Planning}

We first consider a simple grid-world MDP that extends the MDP in Example~\ref{example}, as shown in Fig.~\ref{fig:1}. Starting from initial state $0$, the robot tries to reach and remain in one of the safe cells (labeled with $l_0$ and $l_1$) while avoiding the unsafe cells (labeled with $m$). This task is formally specified by the LTL formula

$$(\textbf{F}\textbf{G} l_0 \vee \textbf{FG} l_1) \wedge \textbf{G} \neg m.$$ 
%
\begin{figure}[t]
    \centering
    \includegraphics[scale=0.2]{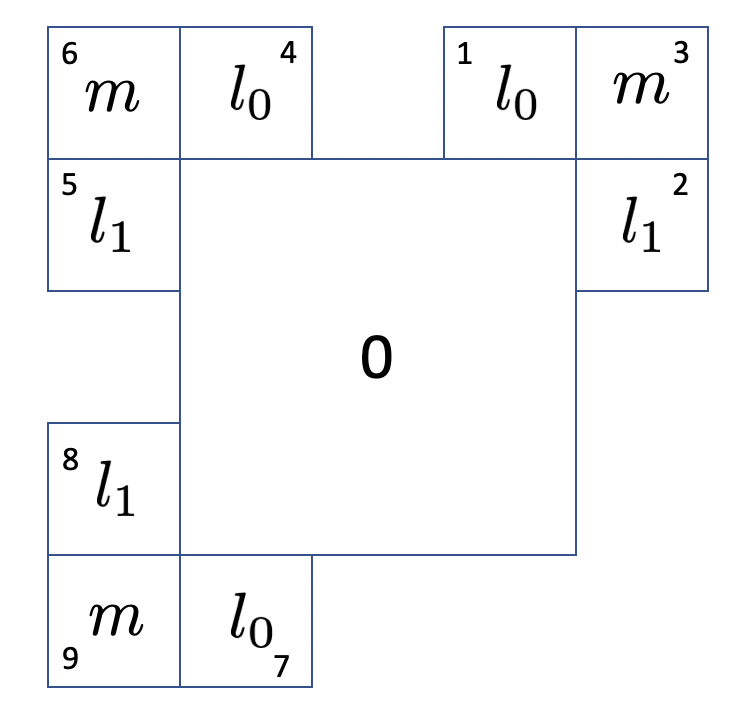}
    \caption{A grid-world MDP case study.}
    \label{fig:1}
\end{figure}
The action $\mathsf{rest}$ which does not change the MDP state is available to the robot in every location. Further, in the initial state $0$, the robot has three more actions available, namely, $\mathsf{ur}$ (upper right), $\mathsf{ul}$ (upper left), and $\mathsf{ll}$ (lower left) such that for $p > 0$,
\begin{equation*}
\begin{aligned}
    & P(0,\mathsf{ur},1) = p, &    P(0,\mathsf{ur},2) = 1 - p,\\
    & P(0,\mathsf{ul},4) = p, &   P(0,\mathsf{ul},5) = 1 - p,\\
    & P(0,\mathsf{ll},7) = p, &   P(0,\mathsf{ll},9) = 1 - p
\end{aligned}
\end{equation*}
In the rest of the cells, the robot can also take action $\mathsf{move}$, by which the robot moves to the vertically adjoining cell with probability $p$ and to the horizontally adjoining cell with probability $1-p$.

Clearly, the deterministic policies which satisfy the LTL formula are of the following form: take an action $a \in \{ \mathsf{ur},\mathsf{ul},\mathsf{ll} \}$ when in state $0$, and action $\mathsf{rest}$ in the other visited states. A reward function $r(s,a)$ gives rise to the expected discounted reward creating an ordering between the above described policies.

The LDBA generated from the LTL formula and the product MDP included 3 states and 30 states, respectively. The resulting MILP, including $173$ continuous and $86$ binary variables required less than $100$-ms runtime for  various values of the parameter $p$, discount factor $\gamma$, and reward function $r(s,a)$. For all choices of the parameters, the generated MILP gives a policy of the form described above, which chooses to move to the quadrant offering  maximum discounted reward.

\subsection{Nursery Scenario}

In this grid-world example, the robot can take actions in the set $\{\mathsf{up, down, left, right}\}$ at each cell of a grid, as shown in Fig.~\ref{fig:2} and Fig.~\ref{fig:3}. On taking each action, the robot moves in the intended direction with probability $0.8$ and sideways with total probability $0.2$, equally divided between the two directions (i.e., a probability of $0.1$ is allocated for moving in each side direction). If the robot is unable to move in a direction due to the presence of a wall, it remains in the same cell. 

An adult, a baby, a charger, and a danger zone are present in four distinct cells of the grid, whose locations are indicated by the truth values of the atomic propositions $a$, $b$, $c$, and $d$,  respectively, associated with each cell of the grid. For example, atomic proposition $a$ is true in cell $(x,y)$ if and only the adult is in cell $(x,y)$.

\begin{figure}[h]
\centering
\begin{minipage}{.25\textwidth}
  \centering
  \includegraphics[width=.42\linewidth]{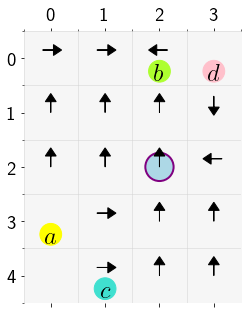}
  \caption{Nursery with equal \\rewards for all actions (A). }
  \label{fig:2}
\end{minipage}%
\begin{minipage}{.25\textwidth}
  \centering
  \includegraphics[width=.42\linewidth]{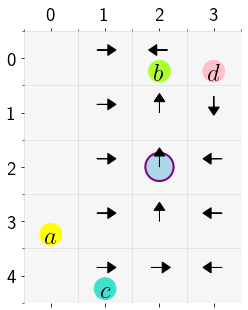}
  \caption{Nursery with lower \\reward for action $\mathsf{up}$ (B).}
  \label{fig:3}
\end{minipage}
\end{figure}

The robot is initially at the charger. Its objective is to check on the baby repeatedly and get back to the charger after doing so, while always avoiding the danger zone. Due to the dynamics of the robot, in some cases, after checking on the baby, the robot remains in the same cell for more than one time step (represented by $b \wedge \textbf{X} b$ in LTL). This behavior  disturbs the baby and requires the robot to notify (visit) the adult. This objective can be formally expressed in LTL as follows:
\begin{equation}
\begin{aligned}
& \varphi := \textbf{G} [\underbrace{\neg d}_\text{($i$)}  \wedge \underbrace{(c \to (\neg a \textbf{U} b))}_\text{($ii$)} \wedge \underbrace{(a \to \textbf{X} (\neg a \textbf{U} b)) }_\text{($iii$)} \wedge \\
  &  \underbrace{((\neg b \wedge \textbf{X} b \wedge \neg \textbf{X} \textbf{X} b) \to (\neg a \textbf{U} c))}_\text{($iv$)} \wedge \underbrace{((b 
    \wedge \textbf{X} b) \to \textbf{F} a)}_\text{($v$)} \wedge\\
   & \underbrace{((b \wedge \neg \textbf{X} b) \to \textbf{X}(\neg b \textbf{U} (a \vee c) ))}_\text{($vi$)} ]
\end{aligned}
\end{equation}
%
%
\noindent Sub-formulae $(i)$-$(vi)$  capture the following specifications which must always be true: $(i)$ avoid the region which is a danger zone; $(ii)$ if charged, visit the baby and do not visit the adult until then; 
$(iii)$ if the adult has been notified, visit the baby and do not visit the adult again until then; $(iv)$ if the baby has been checked on and left undisturbed, get charged and do not notify the adult until then; $(v)$ if the baby has been disturbed, eventually notify the adult; and $(vi)$ on leaving the baby, either notify the adult or charge and do not visit the baby again until then.

Optimal policies were synthesized by solving Problem~\eqref{milp} and simulated for a grid of size $5 \times 4$ and discount factor $0.9$ in two scenarios marked with A and B, respectively. In both scenarios, in the location of the baby, all actions were assigned a reward of $10$. 
In scenario A, all actions are allocated the same reward of $2$ in the rest of the cells. In scenario B, all actions are allocated the same reward ($2$) except for action $\mathsf{up}$ which gains a reward of $1$. 

The LDBA generated from the LTL formula and the product MDP included $57$ states and $1,140$ states, respectively. The resulting MILP, including $9,121$ continuous and $4,560$ binary variables required less than $10$-min runtime. In both cases, the optimal policy was simulated for $300$ time steps. Videos recording the robot trajectories are available online~\cite{vid}. 

Fig.~\ref{fig:2} and~\ref{fig:3} pictorially illustrate the optimal policies for both scenarios in a common state, 
when the robot moves from the charger toward the baby. The arrow directions indicate the optimal action choice. Due to the lower reward for the $\mathsf{up}$ action, the optimal policy in Fig.~\ref{fig:3} exhibits a lower preference for the $\mathsf{up}$ action when compared to the policy in Fig.~\ref{fig:2}, where all actions have equal rewards. Overall, the proposed formulation can effectively discriminate among multiple feasible policies able to almost surely satisfy the LTL formula, and can leverage the total discounted reward to suggest a different course of action in each scenario.

To analyze the scalability of the proposed formulation, we synthesized control policies for nursery scenarios with different grid sizes and reward functions. As reported in Table~\ref{tab:milp}, 
even large grid spaces including $100$ states led to problems that could be solved in less than $90$~minutes. As evident from the last two columns of Table~\ref{tab:milp}, the selection of the reward function may also affect the runtime, with all other parameters left unchanged.

\begin{table}[t]
\caption{MILP Runtime for Different Nursery Scenarios}
\begin{tabular}{c|c|c|c|c|c}
    \hline
     Grid Shape &(5,5)&(6,5)&(8,8)&(10,10)&(10,10)\\ \hline\hline
	\# Constraints & 11400 & 13680 & 29184 & 45600 & 45600 \\ 
 \# Cont. Vars.& 11401 & 13681 & 29185 & 45601 & 45601\\ 
	\# Bin. Vars. & 5700 & 6840 & 14592 & 22800 & 22800\\
	Reward  & 5 & 10 & 10 & 10 & 50\\
	at baby location & &  &  &  & \\
    Runtime [s] & 8.24 & 2.41&17.66&46.45 & 5092 \\ 
\end{tabular}
\label{tab:milp}
\end{table}

\section{Conclusions}

We designed and validated a method to synthesize discounted-reward optimal MDP policies under the constraint that a generic LTL specification is satisfied almost surely. Our approach can account for accumulated reward over the entire trajectory without any assumptions on the reward function. The discounted-reward and LTL objectives are both translated into linear constraints on a  pair of occupancy  measures, which are then connected to provide a single optimal policy via a novel MILP formulation. 
Future plans include the extension of the proposed method to the setting of unknown MDP transitions and maximum LTL satisfaction probabilities that can be less than one.

\bibliographystyle{IEEEtran}
\bibliography{ref}

\end{document}